\documentclass[]{article}

\title{Improved Distributed Algorithms for the Lov\'asz Local Lemma and Edge Coloring}
\author{}
\usepackage{bbm}

\usepackage[T1]{fontenc}
\usepackage[utf8]{inputenc}
\usepackage{lmodern}
\usepackage{fullpage}
\usepackage{url}
\usepackage{amsmath,amsthm,amssymb}
\newtheorem{theorem}{Theorem}
\newtheorem{corollary}[theorem]{Corollary}
\newtheorem{lemma}[theorem]{Lemma}

\newtheorem{definition}[theorem]{Definition}
\newtheorem{notation}[theorem]{Notation}

\usepackage{xspace}
\newenvironment{proofs}{%
	\proof}{\endproof}
\usepackage{cleveref}

\newcommand{\nat}{\ensuremath{\mathbb{N}}}

\usepackage{algorithm}
\usepackage{algpseudocode}

\newcommand{\Prob}[1]{\mathbf{Pr}\left[#1\right]}
\newcommand{\Pru}[2]{\mathbf{Pr}_{#1}\left[#2\right]}
\def\epsilon{\ensuremath{\varepsilon }}
\newcommand{\eps}{\ensuremath{\epsilon }}
\newcommand{\Vars}{\ensuremath{\mathcal V }}
\newcommand{\vars}{\ensuremath{\alpha }}
\newcommand{\ents}{\ensuremath{\mathcal X }}
\newcommand{\rl}{\text{randLOCAL}}
\newcommand{\fin}{{\ensuremath{F}} }

\newcommand{\Part}{\ensuremath{\Phi }}
\newcommand{\aG}{\ensuremath{G_\vars^\ents}}

\newcommand{\Expu}[2]{\mathbf{E}_{#1}\left[#2\right]}

\newcommand{\ca}{{\ensuremath{2.5}} }
\newcommand{\ad}{{\ensuremath{d_\vars}}}
\newcommand{\cb}{{\ensuremath{27}} }
\newcommand{\cc}{{\ensuremath{2.5}} }
\newcommand{\COMMENTED}[1]{{}}

\newcommand{\junk}[1]{\COMMENTED{#1}}

\newcommand{\local}{\textsf{LOCAL}\xspace}
\newcommand{\LOCAL}{\local}

\date{}

\author{Peter Davies\\Durham University\\\url{peter.w.davies@durham.ac.uk}}
\begin{document}

\maketitle

\begin{abstract}
The Lov\'asz Local Lemma is a classic result in probability theory that is often used to prove the existence of combinatorial objects via the probabilistic method. In its simplest form, it states that if we have $n$ `bad events', each of which occurs with probability at most $p$ and is independent of all but $d$ other events, then under certain criteria on $p$ and $d$, all of the bad events can be avoided with positive probability.

While the original proof was existential, there has been much study on the algorithmic Lov\'asz Local Lemma: that is, designing an algorithm which finds an assignment of the underlying random variables such that all the bad events are indeed avoided. Notably, the celebrated result of Moser and Tardos [JACM '10] also implied an efficient distributed algorithm for the problem, running in $O(\log^2 n)$ rounds. For instances with low $d$, this was improved to $O(d^2+\log^{O(1)}\log n)$ by Fischer and Ghaffari [DISC '17], a result that has proven highly important in distributed complexity theory (Chang and Pettie [SICOMP '19]).
	
We give an improved algorithm for the Lov\'asz Local Lemma, providing a trade-off between the strength of the criterion relating $p$ and $d$, and the distributed round complexity. In particular, in the same regime as Fischer and Ghaffari's algorithm, we improve the round complexity to $O(\frac{d}{\log d}+\log^{O(1)}\log n)$. At the other end of the trade-off, we obtain a $\log^{O(1)}\log n$ round complexity for a substantially wider regime than previously known.

As our main application, we also give the first $\log^{O(1)}\log n$-round distributed algorithm for the problem of $\Delta+o(\Delta)$-edge coloring a graph of maximum degree $\Delta$. This is an almost exponential improvement over previous results: no prior $\log^{o(1)} n$-round algorithm was known even for $2\Delta-2$-edge coloring.
\end{abstract}

\section{Introduction}

Our main focus in this paper is on distributed algorithms for the the Lov\'asz Local Lemma.

\subsection{The Lov\'asz Local Lemma}
The Lov\'asz Local Lemma (LLL) is a classic result in probability theory, often used to prove the existence of combinatorial objects by the probabilistic method. Its setup is as follows: 

Consider a set \Vars\ of independent random variables, and a family \ents\ of $n$ (bad) events on these variables; we wish to avoid satisfying any of these events. Each event $A \in\ents $ depends on some subset $\Vars(A) \subseteq \Vars$ of the variables. Define the dependency graph $G^\ents = (\ents , \{\{A, B\} \mid \Vars(A) \cap \Vars(B) \ne \emptyset\})$, i.e. the node set is the set of events, and events are connected by edges if they depend on at least one of the same random variables. Let $d$ denote the maximum degree in this graph: that is, each event $A \in \ents$ shares variables with at most $d$ other events $B \in \ents$. Finally, define $p = \max_{A\in\ents}\Prob{A}$, i.e., an upper bound on the probability of any particular bad event occurring.

The (symmetric\footnote{The symmetric LLL is a special case of the more general aymmetric version, which allows for differing probabilities and dependency degrees between events. In this paper we study only the symmetric LLL, and leave extension to the asymmetric case for future work.}) Lov\'asz Local Lemma then states the following:

\begin{theorem} [Lov\'asz Local Lemma \cite{EL74,Shearer85}]
	If $epd\le 1$, then there exists an assignment of the random variables that avoids all bad events.
\end{theorem}

We will call such an solution a valid assignment. The condition $epd\le 1$ is known as the LLL criterion - this particular version is due to Shearer \cite{Shearer85}, but other criteria are also studied (in particular, criteria with more `slack' often permit faster algorithms for finding the assignment, as we will discuss shortly). Note that the existence of a valid assignment does not depend on the number of bad events $n$, only on the degree of dependence $d$. So, the lemma is useful in situations when $n> 1/p$, in which case simply taking a union bound over all events would not give a positive probability of avoiding all of them.

Since the Lov\'asz Local Lemma is defined on a dependency graph, it makes sense to study the problem of finding a valid assignment as a distributed graph algorithm. We will be concerned with algorithms in the well-known \LOCAL\ model of distributed computing, first introduced by Linial \cite{Linial87}, and this is what we will mean when referring to distributed algorithms (we will not consider other distributed models such as \textsf{CONGEST} in this work).

\subsection{The \LOCAL Model of Distributed Computing}
The model is based on an undirected graph $G=(V,E)$. Nodes (also known as vertices) of the graph are allowed to perform unlimited computation on information they possess; initially they see only their adjacent edges. We will denote $n:=|V|$ to be the number of nodes in the graph, and $\Delta$ to be the maximum node degree.

Algorithms proceed in synchronous communication rounds, in which nodes may send unlimited messages to each of their neighbors. The goal is to minimize the number of communication rounds required to give a correct output at each node.

We will work in the randomized model \rl, in which nodes each have access to their own local stream of random bits. We are concerned with designing algorithms that succeed with high probability (w.h.p.) in $n$, i.e. with probability at least $1-n^{-c}$, for some $c\ge 1$.

\subsection{The Distributed Lov\'asz Local Lemma}
In the \LOCAL model is defined as follows. The underlying input/communication graph is $G^\ents$. This means that it is the bad events that are the nodes performing the computation in the \LOCAL\ model, and nodes can communicate directly iff their corresponding events share dependent variables. We wish to design algorithms to find a valid assignment of variables, under the conditions in which the Lov\'asz Local Lemma guarantees that such an assignment exists. At the beginning of the algorithm, nodes know only the details of their own bad event and dependent variables, and at the end they must output values for all their dependent variables. These values must be consistent (i.e. all nodes for events dependent on a particular variable must output the same value for that variable), and must avoid satisfying any bad event.

We note that the details of how exactly an instance of the LLL is represented in the \LOCAL\ model are not generally very important; there are other reasonable ways to conduct this representation, but all are equivalent up to a constant factor in algorithmic round complexities.

\junk{   The Lov\'asz Local Lemma shows that $\Prob{\cap_{A\in \ents} \bar A} > 0$, under the LLL criterion that $epd \le 1$. Intuitively, if a local union bound is satisfied around each node in $G^\ents$ , with some slack, then there is a positive
probability to avoid all bad events.

The distributed Lov\'asz Local Lemma has emerged as an important meta-problem in distributed graph algorithms. It is arrived at as follows: if a randomized distributed algorithm succeeds at each node with some probability that is
\begin{itemize}
	\item sufficiently high with respect to the maximum degree $\Delta$, but 
	\item not high enough to take a union bound over all $n$ nodes in the graph (i.e., not w.h.p. in $n$),
\end{itemize} then the Lov\'asz Local Lemma can be used to show that with non-zero probability the algorithm does indeed succeed globally - that is, there exists an assignment of each node's input randomness that produces a correct output for every node. This leaves an algorithmic problem: is there an efficient distributed process to \emph{find} such a correct output?}

The distributed Lov\'asz Local Lemma has proven to be key to the study of distributed complexity. For locally-checkable labelling problems (LCLs), on constant-degree graphs, it is the canonical complete problem for the class of problems with $\Theta(\log n)$ deterministic complexity and $\log^{\Theta(1)}\log n$ randomized complexity \cite{CP19}. It also has implications for derandomization \cite{GHK18}, and for connections between the \LOCAL model and descriptive combinatorics \cite{Bernshteyn21}. 

The wide scope of implications for distributed LLL algorithms stem from their utility as meta-algorithms for other problems. This is arrived at as follows: if a randomized distributed algorithm $A$ for some other problem $P$ succeeds at each node with some probability that is
\begin{itemize}
	\item sufficiently high with respect to the maximum degree $\Delta$, but 
	\item not high enough to take a union bound over all $n$ nodes in the graph (i.e., not w.h.p. in $n$),
\end{itemize} then the Lov\'asz Local Lemma can be used to show that with positive probability the algorithm does indeed succeed globally - that is, there exists an assignment of each node's input randomness such that $A$ produces a correct output for $P$ at every node. An LLL algorithm applied to this instance can then find such a valid assignment of input randomness, and simulating $A$ using this randomness solves $P$. In this way, LLL algorithms can be used to amplify the success probability of other algorithms.

\subsection{Distributed Edge Coloring}
Our main application of the LLL result we present will be to edge coloring. Edge coloring is often cited as one of the four classic local distributed graph problems, along with maximal independent set, maximal matching, and vertex coloring (e.g. \cite{PR01,FGK17}). The goal is to assign colors from a palette to all of the edges of the input graph, in such a way that no two edges which share an endpoint are given the same color. 

The size of the available palette greatly affects the difficulty of the problem. A sequential greedy algorithm can color all edges using $2\Delta-1$ colors, since any edge is adjacent to at most $2\Delta-2$ others so would always have a color available however its adjacent edges were colored. This property makes $2\Delta-1$-edge coloring particularly amenable to distributed algorithms. It is also a special case of $\Delta+1$-vertex coloring, since one can consider the problem as vertex coloring the line graph of the original input graph. 

A $poly(\log n)$-round deterministic \local\ algorithm for $2\Delta-1$-edge coloring was found by Fischer, Ghaffari, and Kuhn \cite{FGK17}, which, when combined with randomized algorithms of \cite{EPS15,BEPS16}, also implied a $poly(\log\log n)$-round randomized round complexity. Unlike for $\Delta+1$-vertex coloring, this came prior to the network decomposition result of Rozho\v{n} and Ghaffari \cite{RG20}, which implied $poly(\log n)$-round deterministic algorithms for many problems including the four listed above (a $\Delta+1$-vertex coloring algorithm which does not rely on network decomposition was, however, found later \cite{GK22}).

A further difference from $\Delta+1$-vertex coloring is that, while some graphs (namely cliques and odd cycles) require $\Delta+1$ colors to be properly vertex-colored, edge coloring can potentially be done using far fewer than the $2\Delta-1$ colors required by a greedy algorithm. Vizing's classic theorem \cite{Vizing64} states that all graphs can be edge colored using only $\Delta+1$ colors; however, such colorings are hard to find in a distributed fashion (though a recent paper by Bernshteyn \cite{Bernshteyn22} gives a $poly(\Delta,\log n)$-round deterministic algorithm for doing so).

It remains a major open question to determine how close one can get to the optimal $\Delta+1$ palette size with an efficient distributed edge-coloring algorithm. Here, `efficient' would most commonly mean $poly(\log n)$ rounds for a deterministic algorithm and $poly(\log\log n)$ rounds for a randomized algorithm, as for $2\Delta-1$-edge coloring, though round complexities parameterized by $\Delta$ are also of interest. 

\section{Previous Work}
In this section we discuss prior work on the Lov\'asz Local Lemma and edge coloring.

\subsection{Lov\'asz Local Lemma}
The Lov\'asz Local Lemma was first introduced and proven by Erdős and  Lovász \cite{EL74}. Sequential algorithms to find valid solutions began with Beck \cite{Beck91}, followed by a sequence of improvements \cite{Alon91,CS00,MR98,Moser09}, culminating in the celebrated result of Moser and Tardos \cite{MT10}.

\paragraph{Distributed Algorithms}

Moser and Tardos's result \cite{MT10} also implied the first distributed algorithm for the problem, running in $O(\log^2 n)$ rounds of \rl\ for LLL criterion $ep(d+1)< 1-\eps$ (for any constant $\eps>0$). However, the distributed LLL was not explicitly studied until the work of Chung, Pettie and Su \cite{CPS17}, which, using a similar approach, presented an $O( \log^2 d\log_{1/ep(d+1)} n)$-round algorithm for the criterion $ep(d+1)< 1$ (which was subsequently improved to $O(\log d\log_{1/ep(d+1)} n)$ rounds \cite{Ghaffari16}), and an $O(\log_{1/epd^2} n)$-round algorithm for the criterion $epd^2<1$. 

Taking a different algorithmic approach inspired by that of Molloy and Reed \cite{MR98}, Fischer and Ghaffari \cite{FG17} give improved algorithms for low-degree graphs, which were later extended by Ghaffari, Harris, and Kuhn \cite{GHK18}. An implication of these two works, combined with the subsequent polylogarithmic network decomposition of Rozho\v{n} and Ghaffari \cite{RG20}, is an algorithm for the constructive Lov\'asz Local Lemma taking $O(d^2+\log^{O(1)}\log n)$ rounds, under the criterion $p< d^{-c}$ for some sufficiently large constant $c$. We will call LLL criteria of this type \emph{polynomially-weakened}, following \cite{GHK18}. Note that, for polynomially-weakened criteria, the algorithm of Chung, Pettie and Su \cite{CPS17} requires only $O(\log_{d} n)$ rounds, and so the overall state-of-the-art round complexity is $O(\min\{d^2+\log^{O(1)}\log n,\log_d n\})$.

A lower bound of $\Omega(\log_{d} \log n)$ rounds is known \cite{BF+16}, even for the much weaker criterion $p\le 2^{-d}$, while the best algorithmic round complexity for that criterion remains the $O(\min\{d^2+\log^{O(1)}\log n,\frac{\log n}{d}\})$ given by the algorithms mentioned above \cite{CPS17, FG17,GHK18, RG20}. However, Brandt, Grunau and Rozho\v{n}\cite{BGR20} demonstrated that weakening the criterion any further makes the problem substantially more tractable, giving an $O(d^2+\log^* n)$-round deterministic algorithm for the criterion $p< 2^{-d}$.

\subsection{Edge Coloring}
As mentioned, the ability to color greedily makes $2\Delta-1$ edge coloring particularly amenable to distributed algorithms, and a randomized procedure of Elkin, Pettie and Schneider \cite{EPS15} combined with a deterministic algorithm of Fischer, Ghaffari, and Kuhn \cite{FGK17} yielded the first poly-log-logarithmic randomized round complexity of $O(\log^9 \log n)$ for the problem. This round complexity was later improved to $O(\log^6 \log n)$ \cite{GHK18} and then to $\tilde O(\log^3 \log n)$ \cite{Harris19}. Very recently, a $poly(\log\Delta) + O(\log^* n)$-round deterministic algorithm was also given for the problem \cite{BBKO22}.

With fewer colors, however, the complexity is less straightforward. The leading work is a beautiful result by Chang et al. \cite{CHLPU19}, who give an algorithm for $(1+\eps)\Delta$-edge coloring (for any function $\eps =\omega(\frac{\log^{2.5}\Delta}{\sqrt{\Delta}})$) based on the Lov\'asz Local Lemma, with a round complexity of $O(\log 1/\eps\cdot T_{LLL}(n,d,p)+\log^{O(1)}\log n)$. Here, $T_{LLL}(n,d,p)$ is the time required for a randomized LLL algorithm with parameters $n$, $d=\Delta^{O(1)}$, and $p=exp(-\eps^2\Delta/\log^{4+o(1)}\Delta )$. This bound has a wide variety of implications dependent on the parameters and LLL algorithm used, but we point out three regions of particular interest to us:

\begin{itemize}
	\item When using Fischer and Ghaffari's LLL algorithm \cite{FG17}, the round complexity is  $poly(\Delta, \log\log n)$.
	\item When using Chung, Pettie, and Su's LLL algorithm \cite{CPS17} (or indeed Moser and Tardos's \cite{MT10}), the round complexity is $poly(\log n)$.
	\item When $\eps = \Omega(\frac{\log^3 n}{\sqrt{\Delta}})$, no LLL algorithm is required (since all bad events are avoided with high probability under initial sampling), and the overall round complexity is $O(\log(1/\eps)+\log^* n)$. However, this only improves over a greedy coloring (i.e. uses fewer than $2\Delta-1$ colors) when $\Delta=\Omega(\log^6n)$.\footnote{These illustrative parameters are chosen for clarity rather than optimality, and so this threshold for $\Delta$ can be reduced somewhat, but is still necessarily above $\log n$. }
\end{itemize}

If our aim is an edge coloring with fewer than $2\Delta-1$ colors in $poly(\log\log n)$ \rl\ complexity, results are known only for when $\Delta$ is either $\log^{O(1)}\log n$ (using \cite{FG17}) or $\log^{\Omega(1)} n$ (setting $\eps= \log^{-O(1)} n$ such that no LLL algorithm is needed). That is, there is a range of $\Delta$ between $\log^{\omega(1)}\log n$ and $\log^{O(1)}n$ for which no $\log^{O(1)}\log n$-round algorithm exists for edge coloring using even $2\Delta-2$ colors. It is this range of $\Delta$ for which the distributed LLL is most difficult, and on which this work is focused.

Regarding deterministic algorithms for edge coloring with fewer than $2\Delta-1$ colors, Ghaffari et al. \cite{GKMU18} gave $poly(\log n)$-round deterministic algorithms for $3\Delta/2$ edge coloring, and for $\Delta+o(\Delta)$ edge coloring when $\Delta = \tilde\Omega(\log n)$. As a consequence of \cite{GHK18,RG20}, the algorithm of Chang et al. \cite{CHLPU19} can be derandomized in \local\ at a polylogarithmic overhead, and therefore gives deterministic $(1+\eps)\Delta$-edge coloring (for $\eps =\omega(\frac{\log^{2.5}\Delta}{\sqrt{\Delta}})$) in $poly(\log n)$ rounds (using the LLL algorithm of \cite{CPS17} or \cite{MT10}).

On the lower bound side, Chang et al. \cite{CHLPU19} also show an $\Omega(\log_\Delta \log n)$-round randomized lower bound and an $\Omega(\log_\Delta n)$-round deterministic lower bound for $2\Delta-2$ edge coloring, while for $2\Delta-1$ edge coloring only a $\Omega(\log^* n)$-round lower bound is known \cite{Linial92,Naor91}, further demonstrating the sharp increase in the difficulty of the problem when using fewer than $2\Delta-1$ colors. For an excellent tabular overview on the prior work for both distributed LLL and edge coloring, see Chang et al. \cite{CHLPU19} (though note that the subsequent polylogarithmic network decomposition result \cite{RG20} improved some of the stated bounds).

\section{Our Results and Approach}
In this section we outline our main results, and the ideas needed to attain them.

\subsection{Results}
Our main result is an improved randomized algorithm for the distributed Lov\'asz Local Lemma. In a slightly simplified form, the result is the following:

\begin{theorem}[Simplified version of Theorem \ref{thm:LLL}]\label{thm:LLLsimp}
	There is some constant $c$ such for $1\le r \le \frac{d}{\log d}$, LLL with criterion $p\le 2^{\frac{-cd}{r}}$ can be solved in $O(r+\log^{O(1)}\log n)$ rounds of \rl, succeeding with high probability in $n$.
\end{theorem}

This result is a trade-off between LLL criterion strength and round complexity, and improves round complexities over a wide range of criteria. For the commonly-studied regime of polynomially-weakened LLL criteria, this improves the round complexity from the $O(d^2 + \log^{O(1)}\log n)$ of Fischer and Ghaffari \cite{FG17} to $O(\frac{d}{\log d}+ \log^{O(1)}\log n)$. At the other end of the spectrum, it provides a $\log^{O(1)}\log n$-round algorithm for LLL with the criterion $p\le \min \{d^{-c}, 2^{\frac{-d}{\log^{O(1)}\log n}}\}$ (for some constant $c$), improving substantially over prior results and coming within a polynomial factor of the $\Omega(\log\log n)$-round lower bound of Brandt. et al. \cite{BF+16}.

As our main application, we use our LLL algorithm to improve the complexity of edge coloring using fewer than $2\Delta-1$ colors:

\begin{theorem}\label{thm:edgecolor}
	Let $\eps = \omega(\frac{\log^{2.5}\Delta}{\sqrt{\Delta}})$ be a function of $\Delta$. There is an algorithm for $(1+\eps)\Delta$-edge coloring taking $poly(1/\eps, \log\log n)$ rounds of \rl, succeeding with high probability in $n$.
\end{theorem}

The exact landscape of round complexities for $(1+\eps)\Delta$-edge coloring is complex, with different results taking precedence at different regimes of $n$, $\Delta$, and $\eps$. The most important regime in which Theorem \ref{thm:edgecolor} improves over previous results is when whenever $\Delta$ is between $\log^{\omega(1)}\log n$ and $\log^{1-\Omega(1)} n$, which was previously the hardest case. This improvement implies an efficient randomized algorithm for $\Delta+o(\Delta)$-edge coloring across the whole range of $\Delta$ (as a function of $n$):

\begin{corollary}\label{cor:edge}
	$\Delta+o(\Delta)$-edge coloring can be performed in $\log^{O(1)}\log n$ rounds of \rl, succeeding with high probability in $n$.
\end{corollary}

Previously no such result was known for $\Delta$ between $\log^{\omega(1)}\log n$ and $\log^{1-\Omega(1)} n$. Indeed, for $\Delta = \log^\delta n$, with constant $\delta \in (0,1)$, no prior $\log^{o(1)}n$-round algorithm was known even for $2\Delta-2$-edge coloring.
\subsection{Approach}
We next discuss the new techniques our algorithms employ to attain the improved round complexities.

\paragraph{LLL Algorithm for Resilient Instances}
Let us first consider the LLL algorithm of Fischer and Ghaffari \cite{FG17}, which consists of an $O(d^2+\log^* n)$-round randomized procedure to shatter the graph into small pieces, followed by a $\log^{O(1)}\log n$-round deterministic post-shattering procedure. It proceeds by first computing an $O(d^2)$-vertex coloring of the \emph{square} of the LLL graph; that is, it assigns colors to the vertices of $G^\ents$ such that no vertices within distance $2$ share the same color. Then, the color classes are iterated through. Nodes in the active color class sample their unsampled dependent variables one-by-one (recall that nodes in a distributed LLL instance correspond to bad events, which are dependent on some subset of the set \Vars\ of underlying variables). If sampling a variable causes one of its dependent events to become \emph{dangerous} (informally, too likely to occur), then this is detected by the sampling node, which \emph{reverts} that variable. All the remaining unsampled dependent variables of a dangerous event are \emph{frozen}, meaning they will not be sampled during the randomized procedure and are instead left to the deterministic procedure.\footnote{The version of Fischer and Ghaffari's algorithm described here is the preprint version (\url{https://arxiv.org/abs/1705.04840}), since this version is closest to our own algorithm and provides the best intuition for our changes. However, it contains a minor error in the analysis of events with reverted variables, inherited from \cite{MR98} (see discussion in \cite{PT09}). As a result, the corrected published version of \cite{FG17} no longer reverts variables. We remark, though, that a correct version that still reverts variables is possible, for example by freezing the dependent variables of any event that has a variable reverted, rather than just the dangerous event that caused the variable to be reverted.}

The distance-$2$ coloring ensures that, at any point, each event only has at most one of its dependent variables sampled at a time (since at most one of its neighboring nodes can be active). This allows the analysis to constrain how much damage any particular variable sampling can do, and ensure that the instance always remains satisfiable using the already-sampled values (minus the ones that were reverted). After the randomized procedure, it is shown that most nodes have all of their variables sampled, and the remaining graph (induced by nodes who still have unsampled variables) shatters into small pieces of size $poly(d)\log n$. Then, a deterministic algorithm can be used to fix values for these remaining variables in $\log^{O(1)}\log n$ rounds. We detail this shattering and deterministic process in Section \ref{sec:shatlll}.

We wish to improve over the $O(d^2)$ term, which means we cannot afford to iterate through a distance-$2$ coloring. Instead, we assume we are equipped with a partition of the vertices into fewer than $d^2$ parts (in fact, we will always use at most $O(\frac{d}{\log d})$), and iterate through the parts of that instead. We also assume we have an allocation of variables to one of their dependent events - this event will be responsible for sampling that variable (we can always use an arbitrary allocation, but in some cases one with better properties may be clear). 

The problem with iterating through this partition rather than a distance-$2$ coloring is that nodes no longer have only one of their dependent variables sampled at a time - in fact, all of their neighbors in one part of the partition will sample all of their allocated variables at once. This can cause bad events to rapidly become too likely, especially when variables are reverted. One of the changes we make to combat this is that, rather than reverting individual variables, active events must revert \emph{all} of their allocated variables if one causes a neighboring event to become dangerous. This limits the amount of possible reversion combinations we must consider.

To analyze this, we adapt the definition of an event becoming dangerous, and introduce a property that we call resilience to quantify whether an LLL instance can withstand multiple events sampling and reverting their allocated variables simultaneously. This concept of resilience can be seen as an extension of that of fragility in \cite{GHK18}. The resilience of an LLL instance depends on the event partition chosen. Therefore, to show a result for the general LLL, we must first show how to find good partitions.

We note that the result for resilient instances is stronger than that for the general LLL, and may be of independent interest since in some applications it may give better results when used directly.

\paragraph{General LLL Algorithm}
The property we want from a good partition is simple: we just require that any node has few neighbors in each part of the partition, since again this will help limit the number of possible reversion combinations. We call such a partition a `light partition'. In an interesting `bootstrapping' fashion, we can find such a light partition by framing the problem as a resilient LLL instance, and solving it with our algorithm for resilient LLL.

Our algorithm for general LLL is then simply the algorithm for resilient instances, equipped with a light partition. The difficulty is in proving that general instances are resilient using such a partition (and the parameters of the partition depend on the criterion of the general LLL instance). We do so by essentially taking a union bound over all possible variable reversion combinations, to show that with sufficiently high probability an event would be able to withstand any of them.

\paragraph{Defective Colorings and $\Delta+o(\Delta)$-Edge Coloring}
The first applications we show of our LLL algorithm are for defective colorings. Defective vertex coloring is a classic application of the LLL: it is a relaxation of proper coloring in which nodes are merely required to have few (rather than no) neighbors of their own color. We get improved results for both this and an edge coloring variant, the latter of which is crucial to our main application of (proper) $\Delta+o(\Delta)$-edge coloring.

Specifically, by first employing a defective edge coloring to divide edges into buckets (and also evenly dividing the colors in the palette among these buckets), we can reduce a $(1+\eps)\Delta$-edge coloring instance (where $\eps = \omega(\frac{\log^{2.5}\Delta}{\sqrt{\Delta}})$; this condition comes from \cite{CHLPU19}) into a collection of edge-disjoint instances that can be solved in parallel (since their palettes are also disjoint). These instances have $poly(1/\eps)$ maximum degree, and we use the properties of our defective edge coloring to show that they still have sufficient palette size to be solvable. Then, the main result follows by applying the existing edge coloring algorithm of \cite{CHLPU19}, equipped with our general LLL algorithm (or that of \cite{FG17}), to these smaller instances.

\subsection{Concurrent Work}
Concurrently with our work, Halld{\'o}rsson, Maus and Nolin \cite{HMN22} present results on distributed vertex splitting problems, which are related to the distributed LLL. While their work focuses on the \textsf{CONGEST} model and does not give results for the LLL, some techniques and applications are similar to those here. In particular, they also give a $\log^{O(1)}\log n$-round distributed algorithm for $(1+\eps)\Delta$-edge coloring (though only for constant $\eps > 0$, rather than for $\Delta+o(\Delta)$-edge coloring).

\subsection{Paper Structure}
The structure of the paper is as follows:
\begin{itemize}
	\item In Section \ref{sec:shatlll}, we discuss the shattering framework that is used as the second part of our (and Fischer and Ghaffari's \cite{FG17}) LLL algorithm.
	\item In Section \ref{sec:resLLL}, we define the concept of resilience, and present and analyze our LLL algorithm for resilient instances.
	\item  In Section \ref{sec:partitions}, we define the light partitions that we wish to employ in the general LLL algorithm, and show how to compute them using the LLL algorithm for resilient instances.
	\item  In Section \ref{sec:genLLL} we show how our algorithm, equipped with these light partitions, solves general LLL instances.
	\item  In Section \ref{sec:def} we apply the general LLL result to find improved defective colorings and edge colorings.
	\item In Section \ref{sec:edgecolor} we use these defective edge colorings to solve (proper) $\Delta+o(\Delta)$-edge coloring.
\end{itemize}

\section{Lov\'asz Local Lemma on Shattered Graphs}\label{sec:shatlll}
In this section, we outline the shattering framework that is used for sublogarithmic distributed LLL algorithms. The idea is that, if we can employ a fast randomized process to fix \emph{some} of the variable values, in such a way that the LLL instance remains solvable and the residual LLL graph is shattered into small components, then these remaining small components can be finished off using a deterministic algorithm. This approach is, by now, well understood and utilized in many \LOCAL\ algorithms; the following lemma is implied by a combination of recent results \cite{FG17,GHK18,RG20}, but we formally state it here and sketch a proof for completeness.

\begin{lemma}\label{lem:shatlll}
Consider an LLL instance on $n$ vertices (bad events) with maximum degree $d$. Suppose we have performed some random process which fixes the value of some variables, such that:

\begin{itemize}
	\item the probability that a vertex $v$ does not have all its dependent variables fixed is at most $(ed)^{-4c}$, for some constant $c\ge 1$, and this bound holds even for adversarial choices of the random bits outside the $c$-hop neighborhood of $v$;
	\item conditioned on the values of fixed variables, the probability (over sampling remaining variables from their distributions) that any bad event $v$ is satisfied is at most $\frac{1}{ed^{2.1}}$.	
\end{itemize} 

Then, with high probability (over the randomness of this initial random process), the remaining variables can be fixed in $\log^{O(1)}\log n$ rounds by a deterministic algorithm in such a way that no bad event is satisfied, solving the LLL instance, 
\end{lemma}

\begin{proofs}
By Theorem V.1 of \cite{GHK18}, the induced graph on \emph{residual} vertices (those which do not have all variables fixed) \emph{shatters} into connected components of size at most $O(d^{2c}\log n)$ w.h.p. (and we denote by $N$ such an upper bound on this size). Setting $p=\frac{1}{ed^{2.1}}$, we now have a new LLL instance satisfying $epd^{2.1}\le 1$ on the residual graph. The randomized LLL algorithm of \cite{CPS17} would solve this instance in $O(\log_d N) = O(\log\log n)$ rounds (with each connected component succeeding with high probability in $N$). Furthermore, we can obtain an $(O(\log \log n),O(\log \log n))$-network decomposition of the residual graph in $\log^{O(1)}\log n$ rounds by \cite{RG20} or \cite{GGR21} (this is not entirely straightforward - we must first contract nodes into clusters to reduce the size of connected components to $O(\log n)$, as in \cite{FG17}).

The derandomization framework of \cite{GHK18}, applied to the LLL algorithm of \cite{CPS17} and armed with the $(O(\log n),O(\log n))$-network decomposition, then gives a \emph{deterministic} algorithm to fix the remaining variables, running in $\log^{O(1)}\log n$ rounds.
	
\end{proofs}

For our subsequent LLL algorithms, the challenge is now to provide a fast randomized process meeting the criteria of Lemma \ref{lem:shatlll}; from there, the lemma can be applied to reach the final output.

\section{Lov\'asz Local Lemma on Resilient Instances}\label{sec:resLLL}
In this section, we will present our LLL algorithm (Algorithm \ref{alg:LLL}), and analyze it on instances with a property we will call resilience. This property requires instances to be equipped with a partition $\Part$ of the event set $\ents$ as input. Later on, in Section \ref{sec:genLLL}, we will show how to find partitions for which we can show resilience for the standard general LLL problem, thereby 
turning Algorithm \ref{alg:LLL} into a general LLL algorithm. In some applications, though, stronger results may be obtainable by using the version for resilient instances, and proving resilience for the application directly.

\subsection{Variable Allocation and the Allocated LLL Graph}

To reach our definition of resilience, we must first define concept of an allocation of variables to events. This is a mapping of events to subsets of their dependent variables, in such a way that we partition the entire space of variables.

\begin{definition}
	An allocation of variables $\vars$ is a function $\ents \rightarrow 2^{\Vars}$ such that for all $A \in \ents$, $\vars(A) \subseteq \Vars(A)$, and for all $v \in \Vars$, there is exactly one $A \in \ents$ with $v \in \vars(A)$.
\end{definition}

We phrase an allocation as a function from events to subsets of variables since this will provide cleaner notation in calculations, but it is easier to think of as a mapping of variables to one of their dependent events. The purpose of the allocation is to fix which node (corresponding to an event) is responsible for sampling each variable: our algorithm will only ever allow variables to be sampled by their allocated event.

In the analysis of our subsequent LLL algorithm, we will sometimes consider a restricted version of the $LLL$ graph, which we call the \emph{allocated LLL graph} $\aG$. In this graph, we place an edge between events $A$ to $B$ iff $B$ depends on one of $A$'s \emph{allocated} variables, or vice versa. That is, the edge $\{A,B\}$ is in $E(\aG)$ iff $(\vars(A) \cap \Vars(B))\cup(\vars(B) \cap \Vars(A))  \ne \emptyset$. We will denote by $\ad$ the maximum degree in the graph. Note that that $\ad \le d$, since the edges in $\aG$ are a subset of those in $G^{\ents}$. Furthermore, any pair of adjacent events in $G^{ents}$ are at distance at most $2$ in $\aG$, since they are both adjacent to some event with an allocated variable in $\Vars(A)\cap\Vars(B)$. Therefore, $\ad\le d \le \ad^2+\ad<2\ad^2$.

In the worst case, we can always use an arbitrary variable allocation (allocating variables to any event which is dependent on them), and then $\ad$ can be as large as $d$. However, many applications of the distributed LLL have natural allocations that provide better properties. For example, in vertex coloring problems, both events and variables are associated with particular vertices of the input graph, and so we can allocate each variable to the event for the corresponding input vertex. It therefore transpires to give us stronger bounds for the eventual general LLL result and applications to use $\ad$ rather than $d$.

\subsection{Description of LLL Algorithm}

Our algorithm will work as follows: we proceed in $r$ rounds. In round $i$, the nodes in part $\Part_i$ of the partition will sample their allocated variables. This may make some events \emph{dangerous} (too likely to be satisfied; the formal definition is given in the algorithm). Any node in $\Part_i$ (i.e. that has just sampled its allocated variable values) which is within 1 hop (in $G^{\ents}$) of a newly \emph{dangerous} node becomes \emph{reverted} (joins set $R_{i+1}$). This means that is `undoes' its choice of allocated variable values, leaving those variables to be chosen later by the algorithm for shattered graphs (Lemma \ref{lem:shatlll}). The other nodes in $\Part_i$ are \emph{fixed} (joining set $F_{i+1}$) - their allocated variables will use the sampled values in the ultimate solution. We will refer to the allocated variables of fixed and reverted events as fixed and reverted variables respectively.

Reverting variables in this way can increase the probabilities of events which depend upon these variables. The formal definition of being \emph{dangerous} is designed to ensure that we will still be able to avoid satisfying these events with reverted dependent variables. However, we cannot risk sampling any more of their dependent variables at this stage. So, any node within $2$ hops (in $G^{\ents}$) of a reverted event which has not already sampled its allocated variables (i.e. is in some part $\Part_j$ with $j>i$) becomes \emph{deferred} (joining set $D_{i+1}$). This means that it will \emph{not} sample its allocated variables in round $j$, and these variables will instead be chosen later by the algorithm for shattered graphs (Lemma \ref{lem:shatlll}). We call the allocated variables of deferred events deferred variables. By deferring events within $2$ hops of reverted events, we ensure that events dependent on reverted variables have no further dependent variables sampled in this stage of the algorithm (i.e. until the algorithm for shattered graphs is applied later).

For the purposes of our analysis, we will use the perspective of a \emph{randomness table}: imagine a table with a column for each variable, and in which each row contains a value for that variable, sampled from its distribution. To begin with these all values are hidden; when our algorithm calls for a variable to be sampled, we reveal the next hidden entry in that variable's column. We will only need the first two rows of the table for the analysis of our algorithm. For a set $S$ of variables, we will denote their \emph{first} sampled values (the first row in the randomness table) by $S^1$, and the second by $S^2$. The first sampled values of the allocated variables of \emph{fixed} nodes will be used in the ultimate solution. The second sampled values are just considered for the sake of analysis - all other variables will actually have their values chosen by the algorithm of Lemma \ref{lem:shatlll}.

\subsection{Resilience}

To reason about the probabilities of events being satisfied during the analysis of our algorithm (and also to define the resilience property we use to quantify the hardness of LLL instances), we will need the following notation:

\begin{notation}
	For an event $A$ and a set of events $S\in \ents$, denote by $A_S$ the event that $A$ is satisfied when the allocated variables of events in $S$ (denoted $\vars(S)$) take their \emph{second} sampled value, and all others take their \emph{first}.
\end{notation}
Notice that for any $S$, $\Prob{A_S}=\Prob{A}$; however, specifying which sampled values are taken by variables affects the \emph{correlation} between events. In particular, our analysis will deal with the probabilities of such events conditioned on \emph{some} of the variables having their first values revealed already.

We then define resilience as follows:

\begin{definition}\label{def:resilient}
	Given an LLL instance, equipped with variable allocation \vars\ and event partition $\Part= \{\Part_1, \dots, \Part_r\}$, for each event $A\in\ents$ let $A'$ be the following event (determined by a given assignment to the random variables $\Vars^1$):
	\[A' := \bigcup_{\substack{i \le r,\\S\subseteq \Part_i}}\left\{\Pru{\Vars^2}{A_{S}}\ge d^{-\cc}\right\}\enspace.\]
	
	We call the instance $\Part$-resilient if for each event $A$, 
	\[\Pru{\Vars^1}{A'}  \le d^{-\cb}  \enspace.\]	
	
	We call an instance $r$-resilient, for $r\in \nat$, if there exists some event $r$-partition $\Part$ (i.e. with $|\Part|=r$) for which the instance is $\Part$-resilient.
\end{definition}

Intuitively, $A'$ is the event that, when some set of events $S$ within a single part of the partition revert their allocated variables, $A$ becomes too likely to occur. $A'$ is an event defined on the variable values $\Vars^1$; to evaluate whether $A'$ occurs we need to know these values (but the values $\Vars^2$ are still to be drawn from their distributions). $\Part$-resilience then says, roughly, that each event remains unlikely to occur even if some adversarially chosen subset of the events in any one particular part of $\Part$ revert their allocated variables.

As mentioned, this definition is related to the concept of `fragility' in \cite{GHK18}: in particular, fragility is essentially equivalent in power to $1$-resilience, i.e., when the partition $\Part = \{\Vars\}$, and so the entire variable set is considered at once. As we will show, our extension to use nontrivial partitions and $r>1$ will greatly increase the scope of applications.

\subsection{Statement and Analysis of Algorithm \ref{alg:LLL}}\label{sec:ResAnalysis}

Our algorithm for the Lov\'asz Local Lemma (currently equipped with a partition $\Part$ for which it is resilient, but we will later show how to choose such a partition for the general case) is as follows (Algorithm \ref{alg:LLL}):

\begin{algorithm}[H]
	\caption{LLL(\Part)}
	\label{alg:LLL}
	\begin{algorithmic}
		\State Initialize $\fin_1, R_1, D_1 \gets \emptyset$
		\For{$i = 1$ to $r$}
		\State Initialize $\fin_{i+1} \gets \fin_{i},R_{i+1} \gets R_{i}, D_{i+1} \gets D_{i}$
		\State Each $A\in \Part_i\setminus D_i$ samples allocated variables' first values $\vars(A)^1$ from their distributions
		\State Any $A\in \ents$ is \emph{dangerous} if $\Pru{(\Vars \setminus \vars( \fin_{i}\cup \Part_i\setminus D_i))^1}{A'\mid \vars( \fin_{i}\cup \Part_i\setminus D_i)^1}\ge d^{-\ca}$
		\State Any $A\in \Part_i\setminus D_i$ within one hop of a dangerous event becomes \emph{reverted}; $R_{i+1}\gets R_{i+1}\cup \{A\}$ 
		\State Otherwise, $A\in \Part_i\setminus D_i$ is \emph{fixed}: $\fin_{i+1}\gets \fin_{i+1} \cup \{A\}$
		\State All $A\in \Part_{j}$ for $j>i$ within $2$ hops of a reverted event are \emph{deferred}: $D_{i+1}\gets D_{i+1} \cup \{A\}$.
		
		\EndFor
		
		\State Variables $\vars(F_r)$ are fixed as their first sampled values $\vars(F_r)^1$
		\State Remaining variable values are chosen using algorithm for shattered graphs (Lemma \ref{lem:shatlll})
		
	\end{algorithmic}
\end{algorithm}

It is clear that upon completion of the first stage of the algorithm (by which we mean all lines except the final call to Lemma \ref{lem:shatlll}), all events are either fixed, reverted, or deferred. This first stage also clearly takes $O(r)$ rounds in \LOCAL. It therefore remains to show that this first stage meets the conditions of Lemma \ref{lem:shatlll}. Then, Lemma \ref{lem:shatlll} will fix the remaining variable values to reach a valid solution in $O(r+\log^{O(1)}\log n)$ total rounds.

We analyze the algorithm as follows: fix an event $A$ to consider. By our assumption of resilience, we have $\Pru{\Vars^1}{A'}  \le d^{-\cb}$ (where $A'$ is as in the definition of resilience). 

In the following proofs, we will use $N(A)$ to denote the inclusive neighborhood of event $A$ in $G^\ents$, that is:
\[N(A) := \{A\}\cup  \{B\in \ents: \{A,B\}\in E(G^\ents)\}\enspace.\]
Similarly, we denote by $N_\vars(A)$ the inclusive neighborhood of $A$ in $\aG$.

For each $i$, let $E_i$ be the event that $\mathbf{Pr}_{(\Vars\setminus \vars(\fin_i))^1}[A'\mid \vars(\fin_i)^1]\ge d^{-\ca}$, and let $E=\cup_{i\le r} E_i $. Intuitively, $E_i$ is the event that, conditioned on the fixed variable values up to round $i$, $A'$ becomes too likely. $E$ is then the event that $A'$ becomes too likely at any point during the algorithm. We wish to upper-bound the probability of $E$:

\begin{lemma}\label{lem:safe}
$\Prob{E}\le d^{-24.5}$. 
	
\end{lemma}

\begin{proof}
	We denote $E^*_j$ to be the event that $E_j$ is the \emph{first} of the $E_i$ to occur, i.e. $E^*_j = E_j \cap \bigcap_{i<j}\bar{E_i}$ (where $\bar{E_i}$ denotes the complement of $E_i$). Then, we can rewrite $E$ as the \emph{disjoint} union of the $E^*_j$ ($E= \dot\bigcup_{j\le r} E^*_j$). Consequently,

	\begin{align*}
		\mathbf{Pr}[E] &= \sum_{j\le r} \Prob{E^*_j} =
		\sum_{j\le r} \frac{\mathbf{Pr}[A'\cap E^*_j]}{\Prob{A'\mid E^*_j}}
		\le 
		\sum_{j\le r} \frac{\mathbf{Pr}[A'\cap E^*_j]}{d^{-\ca}}\\
		&=d^{\ca}\sum_{j\le r} \mathbf{Pr}[A'\cap E^*_j]
		\le d^{\ca} \Prob{A'} \le d^{\ca-\cb}=d^{-24.5}\enspace.
	\end{align*}

(Here, for ease of notation, the probabilities are over all $\Vars^1$, but note that the events $E_i$ and $E^*_i$ depend only on the values of $\vars(F_i)^1$, and treat the remaining variables as unfixed.)
\end{proof}

We can use this bound to upper-bound the probability of an event becoming dangerous, and thereby also upper-bound the probability that an event has any of its dependent variables frozen or reverted.

\begin{lemma}\label{lem:Eprob}
	For any event $A$, the probability $A$ has any deferred or reverted variables in $\Vars(A)$ by the end of Algorithm \ref{alg:LLL} is at most $2d^{-20.5}$, even if random choices outside $A$'s $5$-hop neighborhood in $G^{\ents}$ are chosen adversarially.
\end{lemma}

\begin{proof}
	By Lemma \ref{lem:safe}, with probability at least $1-d^{-24.5}$, $E$ does not hold. We show that in this case, $A$ cannot become dangerous during the algorithm. 
	
	Assume, for the sake of contradiction, that $A$ first becomes dangerous in round $i$. Then, if none of $A$'s dependent variables have been reverted prior to round $i$, we have $\fin_{i+1} \cap N(A) =(\fin_{i}\cup \Part_i\setminus D_i)\cap N(A)$, and so $\bar{E}_{i+1}$ implies precisely that $A$ does not become dangerous. If, on the other hand, some of $A$'s dependent variables \emph{have} been reverted, then all of $A$'s remaining dependent variables are deferred, so in round $i$ no variable in $\Vars(A)$ is sampled and $A$ therefore cannot become dangerous.
	
	So, any event $A$ becomes dangerous with probability at most $d^{-24.5}$, even if randomness outside its $1$-hop neighborhood is chosen adversarially (since the occurrence of $E$ is dependent only on the sampled values of $\Vars(A)$). Dangerous events cause events up to $4$-hops away\footnote{The $4$-hop distance here is because dangerous events cause events $1$-hop away to revert, reverted events cause events $2$-hops away to defer, and events $1$-hop away from deferred events can therefore have deferred dependent variables.} (of which there are fewer than $2d^4$) to have deferred dependent variables, and so the probability that $A$ has any deferred dependent variables is at most $2d^{-20.5}$, even if randomness outside its $5$-hop neighborhood is adversarial. 
\end{proof}

Lemma \ref{lem:Eprob} will be sufficient to give us one of the important properties we need to apply Lemma \ref{lem:shatlll}: most events will have all their dependent variables fixed by the first stage of the algorithm, and will no longer need to participate. The induced graph on events that still have unfixed variables will shatter into small pieces, as we desired. The remaining thing to show before we can apply Lemma \ref{lem:shatlll} is that the residual LLL instance remaining after this shattering prcess is still solvable:

\begin{lemma}\label{lem:p'}
	After Algorithm \ref{alg:LLL}, i.e. after fixing the values for $\vars(\fin_{r+1}\cup R_{r+1})^1$, all events $A$ have \[\Pru{\Vars\setminus \vars(\fin_{r+1}\cup R_{r+1})^1}{A_{R_{r+1}}|\vars(\fin_{r+1}\cup R_{r+1})^1}\le 2d^{-\ca}\enspace.\]
\end{lemma}

\begin{proof}	
	We analyze three possible cases:
	
	\paragraph{Case 1: $A$ has no reverted dependent variables.}
	In this case, we have \[(\fin_{r+1}\cup R_{r+1}) \cap N(A)= \fin_{r+1} \cap N(A) = (\fin_{r}\cup \Part_r \setminus D_r) \cap N(A).\] Since $A$ was not dangerous in round $r$, we have \[\Pru{(\Vars \setminus \vars( \fin_{r}\cup \Part_r\setminus D_r))^1}{A'\mid \vars( \fin_{r}\cup \Part_r\setminus D_r)^1}< d^{-\ca},\] i.e., \[\Pru{(\Vars \setminus \vars(\fin_{r+1}\cup R_{r+1}))^1}{A'\mid \vars(\fin_{r+1}\cup R_{r+1})^1}<d^{-\ca}.\]
	
	From the definition of $A'$, it is a superset of the event $\left\{\Pru{\Vars^2}{A_{R_{r+1}}}\ge d^{-\cc}\right\}$. So, 
	
\[\Pru{(\Vars \setminus \vars(\fin_{r+1}\cup R_{r+1}))^1}{\left\{\Pru{\Vars^2}{A_{R_{r+1}}}\ge d^{-\cc}\right\}\mid \vars(\fin_{r+1}\cup R_{r+1})^1}<d^{-\ca}.\]

Then,

\begin{align*}
\Pru{\Vars\setminus \vars(\fin_{r+1}\cup R_{r+1})^1}{A_{R_{r+1}}|\vars(\fin_{r+1}\cup R_{r+1})^1}&= 
\Expu{(\Vars \setminus \vars(\fin_{r+1}\cup R_{r+1}))^1}{\Pru{\Vars^2}{A_{R_{r+1}}}|\vars(\fin_{r+1}\cup R_{r+1})^1}\\
&\hspace{-2in}\le 
\Pru{(\Vars \setminus \vars(\fin_{r+1}\cup R_{r+1}))^1}{\left\{\Pru{\Vars^2}{A_{R_{r+1}}}\ge d^{-\cc}\right\}|\vars(\fin_{r+1}\cup R_{r+1})^1} + d^{-\cc} < 2d^{-\cc}\enspace.
\end{align*}
\paragraph{Case 2: $A$ has reverted dependent variables but did not become dangerous.}

Let $i$ be the round in which $A$ had dependent variables reverted (this can only happen in at most one round, since subsequently all of $A$'s remaining dependent variables are deferred). All of $A$'s dependent variables that were sampled in round $i$ were either fixed or reverted, and all \emph{subsequent} dependent variables were deferred, so we have \[ ( \fin_{i}\cup \Part_i\setminus D_i)\cap N(A)=(\fin_{r+1}\cup R_{r+1}) \cap N(A)\enspace .\]

Since $A$ did not become dangerous in round $i$, we have \[\Pru{(\Vars \setminus \vars( \fin_{i}\cup \Part_i\setminus D_i))^1}{A'\mid \vars( \fin_{i}\cup \Part_i\setminus D_i)^1}< d^{-\ca},\] i.e., 
\[\Pru{(\Vars \setminus \vars(\fin_{r+1}\cup R_{r+1}))^1}{A'\mid \vars(\fin_{r+1}\cup R_{r+1})^1}< d^{-\ca}.\]

From this point we can follow an identical argument to Case 1 to reach:
	
\begin{align*}
	\Pru{\Vars\setminus \vars(\fin_{r+1}\cup R_{r+1})^1}{A_{R_{r+1}}|\vars(\fin_{r+1}\cup R_{r+1})^1}&
 < 2d^{-\cc}\enspace.
\end{align*}
	\paragraph{Case 3: $A$ became dangerous.}
	Let $i$ be the round in which $A$ became dangerous. We have \[(\fin_{r+1}\cup R_{r+1}) \cap N(A)=  (\fin_{i}\cup \Part_i \setminus D_i) \cap N(A).\]
	 Since $A$ was not dangerous in round $i-1$, we have 
	\[\Pru{(\Vars \setminus \vars( \fin_{i-1}\cup \Part_{i-1}\setminus D_{i-1}))^1}{A'\mid \vars( \fin_{i-1}\cup \Part_{i-1}\setminus D_{i-1})^1}< d^{-\ca}.\]
Since $A$ became dangerous in round $i$, it must have had no reverted neighbors in round $i-1$ (since this would have resulted in all of its remaining dependent variables being deferred, and none would have been sampled in round $i$ to make $A$ dangerous). Furthermore, no neighbors of $A$ are fixed in rounds $i$ onwards: they are all either reverted or deferred. So, $( \fin_{i-1}\cup \Part_{i-1}\setminus D_{i-1}) \cap N(A) = \fin_{i} \cap N(A)  = \fin_{r+1} \cap N(A)$, i.e., 
	\[\Pru{(\Vars \setminus \vars( \fin_{r+1}))^1}{A'\mid \vars( \fin_{r+1})^1}< d^{-\ca}.\]	
	
As before, $A'$ is a superset of the event $\left\{\Pru{\Vars^2}{A_{R_{r+1}}}\ge d^{-\cc}\right\}$. So, 
	
	\[\Pru{(\Vars \setminus \vars( \fin_{r+1}))^1}{\left\{\Pru{\Vars^2}{A_{R_{r+1}}}\ge d^{-\cc}\right\}\mid \vars( \fin_{r+1})^1}< d^{-\ca}.\]	
	
By its definition, the event $A_{R_{r+1}}$ is entirely independent of the values of $\vars(R_{r+1})^1$ (it depends only on the \emph{second} sampled values of $\vars(R_{r+1})$, and the first sampled values of the other variables). So, 

	\[\Pru{(\Vars \setminus \vars(\fin_{r+1}\cup R_{r+1}))^1}{\left\{\Pru{\Vars^2}{A_{R_{r+1}}}\ge d^{-\cc}\right\}\mid \vars(\fin_{r+1}\cup R_{r+1})^1}< d^{-\ca}.\]	
	
As in the final step of Case 1, we therefore reach:

\begin{align*}
	\Pru{\Vars\setminus \vars(\fin_{r+1}\cup R_{r+1})^1}{A_{R_{r+1}}|\vars(\fin_{r+1}\cup R_{r+1})^1}&
	< 2d^{-\cc}\enspace.
\end{align*}
	
\end{proof}

This is sufficient to show that the residual post-shattering LLL instance is solvable, and thereby complete the analysis for Algorithm \ref{alg:LLL}:

\begin{theorem}\label{thm:reslll}
	Any $r$-resilient LLL instance, provided with an $r$-partition $\Part$ for which it is $\Part$-resilient, can be solved in $O(r+\log^{O(1)}\log n)$ rounds in \rl, succeeding with high probability.
\end{theorem}

\begin{proof}
The first stage of Algorithm \ref{alg:LLL}, taking $O(r)$ rounds, finds the first sampled values of some of the variables, specifically, $\vars(\fin_{r+1}\cup R_{r+1})^1$. For $\vars(\fin_{r+1})$, we fix these as the final values; the other variables remain unfixed. We now have a new LLL instance, consisting of only the unfixed variables, and the events which are dependent on them.
	
By Lemma \ref{lem:safe}, the probability that any event $A$ does not have all its dependent variables fixed is at most $2d^{-20.5}< (ed)^{-20}$, even if random choices outside $A$'s $5$-hop neighborhood are chosen adversarially. By Lemma \ref{lem:p'}, the probability (over sampling unfixed variables) of each event $A$ occuring is at most $2d^{-\ca} < \frac{1}{ed^{2.1}}$.
	
So, the conditions for Lemma \ref{lem:shatlll} are met, and the call to the corresponding algorithm fixes the remaining variables in such a way that no bad event is satisfied, in $\log^{O(1)}\log n$ rounds, with high probability.
\end{proof}

\section{Light Partitions}\label{sec:partitions}
To effectively use Theorem \ref{thm:reslll} for applications, we must show a variable allocation and event partition such that we reach an $r$-resilient LLL instance, for as low an $r$ as possible. The choice of allocation is generally clear from the application; in this section, we will show how to obtain a good partition.

The property we want for our partition is that each node has few neighbors in each part of the partition. In our LLL algorithm, this will mean that a node can have only few of its variables reverted, which will make it easier to satisfy resilience. The formal definition we will use is the following:

\begin{definition}
	An $x$-light partition of a graph $G$ of maximum degree $\Delta$ is a partition of nodes into $\frac{\Delta}{x}$ parts such that each node has $O(x)$ neighbors in each part. 
\end{definition}

This definition is very similar to the concept of a \emph{frugal coloring}, a classic application of the LLL introduced by Hind, Molloy and Reed \cite{HMR97}. The only difference is that a frugal coloring is, under most definitions, required to be a proper coloring (i.e. to have no monochromatic edges), which we do not require from a light partition. This weakening of the definition is important, since frugal colorings using fewer than $\Delta$ colors do not, in general, exist, while we will be concerned with light partitions using $o(\Delta)$ parts.

To find such a partition, we will use our LLL algorithm for resilient instances, equipped with the trivial $1$-partition (all events in the same part). The resulting light partitions can then be used for our subsequent applications (including the general LLL result), in a `bootstrapping' fashion. 

We will prove the following lemma:

\begin{lemma}\label{lem:frugal1}
A $\log \Delta$-light partition  can be found in $\log^{O(1)}\log n$ rounds in \rl, succeeding with high probability.
\end{lemma}

We will show how to formulate the problem as a $1$-resilient LLL instance; the lemma then follows from Theorem \ref{thm:reslll}.\footnote{In fact, since $1$-resilience is equivalent in power to fragility from \cite{GHK18}, one could alternatively employ Theorem I.7 of \cite{GHK18}. However, this proof will serve as a useful introductory `warm-up’ to the use of Theorem \ref{thm:reslll}, which we will need for all of our later applications since they will be $r$-resilient only for $r=\omega(1)$.}
The LLL instance will be based on the following very simple random process: each vertex picks one of the $\Delta/\log\Delta$ parts uniformly at random. Consequently, our LLL graph $G^{\ents}$ is constructed as follows: 

\begin{itemize}
	\item The vertex set corresponds to that of the input graph ($\ents = \{_vA:v\in V\}$).
	\item The set of variables $\Vars$ will consist of one variable $part(v)$ for each $v\in V$, and this variable will be allocated to $_vA$ (i.e. $\vars(_vA) = \{part(v)\}$).
	\item These variables $part(v)$ are each uniformly distributed in $[\Delta/\log\Delta]$.
	\item The \emph{bad event} $_vA$ is that $v$ has more than $99\log \Delta$ neighbors in some part (so avoiding all bad events would imply a $\log\Delta$-light partition).
	\item We therefore have the  edge $\{_vA,\ _wA\}$ in $E(\aG)$ iff edge $\{v,w\}$ is in $E(V)$, i.e. $\ad = \Delta$.
	\item Events $_vA\ne$ $_wA$ are adjacent in $G^{\ents}$ iff $dist_G(v,w)\le 2$ (i.e. $d< 2\Delta^2$).
	\item The trivial $1$-partition $\Part$ is $\{\{\ents\}\}$.
\end{itemize}  

We will show that this LLL instance is $1$-resilient. To do so, we must meet the following definition (Definition \ref{def:resilient}, simplified for $r=1$):

\begin{definition}
	Given an LLL instance, equipped with variable allocation \vars\, for each event $A\in\ents$ let $A'$ be the event that 
	\[\bigcup_{S\subseteq \ents}\left\{\Pru{\Vars^2}{A_{\vars(S)}}\ge d^{-\cc}\right\} \enspace.\]
	
	We call the instance $1$-resilient if for each event $A$, 
	\[\Pru{\Vars^1}{A'\mid \Vars^1}  \le d^{-\cb}  \enspace.\]	
\end{definition}

For each $_vA$, we define $_vA^*$ to be the event that $v$ has at most $49\log\Delta$ neighbors in some part, under the values (part assignment) $\Vars^1$.

We will now prove an upper bound on the probabilities of the $_vA^*$, and then prove that $_vA^*$ is a superset of $_vA'$ (and hence the upper bound on probability also applies to the $_vA'$). 

\begin{lemma}\label{lem:astar}
	For any $v$, $\Pru{\Vars^1}{_vA^*}< d^{-\cb}$.
\end{lemma}
\begin{proof}
	For each part $i$ in $[\Delta/\log\Delta]$, denote by $e(i)$ the number of neighbors of $v$ taking part $i$ under $\Vars^1$, and by $\mu$ its expectation: we have $\mu\le \log\Delta$. By a standard Chernoff bound, for $\delta \ge 1$,
	
	\[\Prob{e(i)\ge \delta\mu}
	\le e^{(3-\delta)\mu}\enspace.
	\] 
	
	We will set $\delta	= 49\log\Delta/\mu $.  Then (assuming $\Delta$ is at least a sufficiently large constant, since otherwise the lemma is trivial):
	
	\begin{align*}
		\Prob{e(i)\ge 49\log\Delta}
		&\le e^{3\mu-49\log\Delta}\le e^{-46\log\Delta}<\Delta^{-66}<d^{-30}\enspace.
	\end{align*}
	
	Taking a union bound over all $\Delta/\log\Delta$ parts, $\Pru{\Vars^1}{_vA^*}< d^{-\cb}$.
\end{proof}

\begin{lemma}\label{lem:a'astar}
	For any $A$, $\Pru{\Vars^1}{A'\mid \Vars^1}< d^{-\cb}$.
\end{lemma}
\begin{proof}
	By the same argument as above, the probability that $v$ has at least  $49\log\Delta$  neighbors in any one part under variables' \emph{second} sampled values $\Vars^2$ is at most $d^{-\cb}$. If this does not occur, then for any set $S$, replacing the values $\vars(S)^1$ by $\vars(S)^2$ cannot increase the number of $v$'s  neighbors in a particular part by more than $49\log\Delta$.
	
	Now, if we take any assignment of $\Vars^1$ which satisfies $_vA'$, we know (by definition of $_vA'$) that for any $S\subseteq \ents$, $\Pru{\Vars^2}{_vA_{\vars(S)}}\ge d^{-\cc}$. So, we know that under $\Vars^1$, $v$ has at least $99\log\Delta - 49\log\Delta = 50\log\Delta$ neighbors in some part (since otherwise we would have $\Pru{\Vars^2}{_vA_{\vars(S)}}< d^{-\cb}$). Then, $\Vars^1$ satisfies $A^*$, so $A' \subseteq A^*$. This means that $\Pru{\Vars^1}{A'}\le \Pru{\Vars^1}{A^*} < d^{-\cb}$ by Lemma \ref{lem:astar}.
\end{proof}

We have now met the condition of $1$-resilience, which is sufficient to prove Lemma \ref{lem:frugal1}:

\begin{proof}[Proof of Lemma \ref{lem:frugal1}]
Our LLL instance for $\log \Delta$-light partition is $1$-resilient, so by Theorem \ref{thm:reslll} can be solved in $\log^{O(1)} \log m$ rounds of \rl\ (succeeding with high probability in $n$).
\end{proof}

Note that a $\log\Delta$-light partition implies an $x$-light partition for any $x\ge \log \Delta$:

\begin{corollary}\label{cor:frugal1}
For any $x\ge \log\Delta$, an $x$-light partition  can be computed in $\log^{O(1)}\log n$ rounds in \rl, succeeding with high probability.
\end{corollary}

\begin{proof}
Compute a $\log\Delta$-light partition using Lemma \ref{lem:frugal1}, and then lexicographically group the $\frac{\Delta}{\log\Delta}$ parts into $\frac{\Delta}{x}$ superparts as equally as possible (i.e., with either $\lfloor\frac{x}{\log\Delta}\rfloor$ or $\lceil\frac{x}{\log\Delta}\rceil$ parts in each superpart). Each vertex now has at most $O(\log\Delta) \cdot \lceil\frac{x}{\log\Delta}\rceil = O(x)$ neighbors in each superpart, constituting an $x$-light partition.
\end{proof}

\section{General Lov\'asz Local Lemma}\label{sec:genLLL}
Recall that in general LLL instances, we are given an LLL graph $G^{\ents}$ as our input (and communication graph), in which each node represents a bad event which occurs with probability at most $p$, and events sharing dependent variables are joined with edges. The maximum degree in $G^{\ents}$ is denoted $d$.

In our case, we also assume that we have a variable allocation $\vars$ - a specific allocation with good properties may be apparent from applications; if not, we can take an arbitrary allocation. This defines the allocated LLL graph $\aG$, a subgraph of $G^{\ents}$, with maximum degree at most $\ad$ (and recall further that $\ad\le d <2\ad^2$). 

In this section we prove the following result for the general LLL:

\begin{theorem}\label{thm:LLL}
There is some constant $c$ such for $1\le r \le \frac{\ad}{\log \ad}$, LLL with criterion $p\le 2^{\frac{-c\ad}{r}}$ can be solved in $O(r+\log^{O(1)}\log n)$ rounds of \rl, succeeding with high probability in $n$.
\end{theorem}

Note that this statement involves $\ad$ rather than $d$ (as in the statement of Theorem \ref{thm:LLLsimp}); since $\ad$ can be significantly lower than $d$ in some applications, this makes the statement stronger (and is crucial in our applications to defective coloring and edge coloring later). In cases where one does not have an allocation with good properties, and so $\ad$ is not significantly lower than $d$, Theorem \ref{thm:LLLsimp} can be used for simplicity. We may assume throughout that $d$ and $\ad$ are at least sufficiently large constants, since otherwise the theorem follows from Fischer and Ghaffari \cite{FG17}. 

\begin{proof}

Our aim will be to show that an LLL instance with criterion $p\le 2^{\frac{-c\ad}{r}}$ is $r$-resilient. To do so, we will first compute a $\frac{\ad}{r}$-light partition of $\aG$ to use as our event partition, which can be done in $O(\log^{O(1)}\log n)$ rounds by Corollary \ref{cor:frugal1}. $\Part$ therefore has $r$ parts, and any event $A$ has $O(\frac{\ad}{r})$ neighbors in any $\Part_i$. We will need to reason about the constant in this $O()$ notation, so let $\gamma$ be a sufficiently large constant that any $A$ has fewer than $\frac{\gamma\ad}{r}$ neighbors (in $\aG$) from any $\Part_i$, i.e. \[\forall A\in \ents, i \in [r],: |N_\vars(A)\cap \Part_i| \le \frac{\gamma\ad}{r}\enspace.\]

Recall that, for any $A$, $A'$ is defined to be the event:

\[\bigcup_{\substack{i \le r,\\S\subseteq \Part_i}}\left\{\Pru{\Vars^2}{A_{\vars(S)}}\ge d^{-\cc}\right\}\enspace.\]

We must now show that the instance is $r$-resilient, i.e., for each event $A$, 
$\Pru{\Vars^1}{A'\mid \Vars^1}  \le d^{-\cb} $.

To do so, we simply take a union bound over all $i$ and $S\cap N_\vars(A)$: 

\begin{align*}
\Pru{\Vars^1}{\bigcup_{\substack{i \le r,\\S\subseteq \Part_i}}\left\{\Pru{\Vars^2}{A_{\vars(S)}}\ge d^{-\cc}\right\}\mid \Vars^1}
&\le \sum_{\substack{i \le r,\\S\subseteq (\Part_i\cap N_\vars(A))}}\Pru{\Vars^1}{\left\{\Pru{\Vars^2}{A_{\vars(S)}}\ge d^{-\cc}\right\}\mid \Vars^1}\\
&\le \sum_{\substack{i \le r,\\S\subseteq (\Part_i\cap N_\vars(A))}}\frac{\Prob{A}}{d^{-\cc}}\\
&\le r\cdot 2^{\frac{\gamma\ad}{r}} \cdot  pd^{\cc}\\
&\le \ad \cdot 2^{\frac{\gamma\ad}{r}} \cdot  2^{\frac{-c\ad}{r}}d^{\cc}\enspace.
\end{align*}

We set $c=\gamma+80$, and use $d<2\ad^2$:

\begin{align*}
\Pru{\Vars^1}{A'\mid \Vars^1}  &\le \ad  \cdot  2^{\frac{-80\ad}{r}}d^{\cc}
\le \ad  \cdot  \ad^{-80}d^{\cc}
= \ad^{-79}d^{\cc}\le d^{-\cb}\enspace.
\end{align*}

We have thus proven $r$-resilience. The theorem then follows by Theorem \ref{thm:reslll}.

\end{proof}

Theorem \ref{thm:LLL} gives us a trade-off between the LLL criterion and the round complexity. To illustrate the two end-points of this trade-off, we give the following corollaries:

\begin{corollary}
There is some constant $c$ such that LLL with $p
\le d^{-c}$ can be solved in $O(\frac{\ad}{\log \ad} +\log^{O(1)}\log n)$ rounds of \rl, succeeding with high probability.\qed
\end{corollary}

This result is for LLL with polynomially-weakened LLL criterion, and improves the $O(d^2+\log^{O(1)}\log n)$ bound of Fischer and Ghaffari's algorithm \cite{FG17,GHK18,RG20} to $O(\frac{d}{\log d}+\log^{O(1)}\log n)$ (with worst-case allocation; the improvement may be larger if $\ad<d$). Note also that Chung, Pettie, and Su's LLL algorithm \cite{CPS17} solves such instances in $O(\log_d n)$ rounds, so the overall complexity upper bound is now $O(\min\{\frac{\ad}{\log \ad} +\log^{O(1)}\log n,\frac{\log n}{\log \log n}\})$. 

\begin{corollary}
There is some constant $c$ such that LLL with $p
	\le \min \{\ad^{-c}, 2^{\frac{-\ad}{\log^{O(1)}\log n}}\}$ can be solved in $\log^{O(1)}\log n$ rounds of \rl, succeeding with high probability.
\end{corollary}

This result gives $\log^{O(1)}\log n$-round LLL for a substantially wider regime than previously known. Previously, such a round complexity required $d = \log^{O(1)}\log n$ \cite{FG17,GHK18,RG20} or $p \le 2^{-\frac{\log n}{\log^{O(1)}\log n}}$ \cite{CPS17} (and note that once $p\le n^{-2}$, LLL is trivially $0$-round solvable). Furthermore, this round complexity is within a polynomial factor of optimality, since the lower bound of Brandt et al. \cite{BF+16} demonstrates that LLL requires $\Omega(\log\log n)$ rounds of \rl\ even in constant-degree graphs and under the weaker criterion $p\le 2^d$.

\section{Defective Vertex and Edge Coloring}\label{sec:def}

Our first application of Theorem \ref{thm:LLL} will be for defective (vertex) coloring. We will then apply a very similar argument to an edge coloring variant of the problem, which will be used in our final application to $\Delta+o(\Delta)$ edge coloring in Section \ref{sec:edgecolor}.

\subsection{Defective Vertex Coloring}
Defective coloring is a relaxation of proper coloring, in which vertices are required to have few (rather than no) neighbors of the same color as themselves. The specific version we will solve is the following:

\begin{definition}
	An $(x,q)$-defective coloring of a graph $G$ with maximum degree $\Delta$ is a $\frac{\Delta}{x}$-coloring of nodes such that each node has fewer than $x + x/q$ neighbors of its own color. 
\end{definition}

This definition is phrased in a slightly non-standard way, for ease of notation in our applications. Our $(x,q)$-defective coloring definition corresponds to a `$\frac{\Delta}{x}$-coloring of defect $x + x/q$' in more standard terms.

We will first find a defective coloring with two colors, and then iterate to increase the number of colors:

\begin{lemma}\label{lem:def1}
For any $q=o(\sqrt\frac{\Delta}{\log^3\Delta})$, a $(\frac{\Delta}{2},2q\log\Delta)$-defective coloring can be found in $O((q \log \Delta)^2 + \log^{O(1)}\log n)$ rounds of \rl, succeeding with high probability in $n$.
\end{lemma}

\begin{proof}
A $(\frac{\Delta}{2},q\log\Delta)$-defective coloring is a coloring using only $2$ colors, such that every node has fewer than $\frac{\Delta}{2}+\frac{\Delta}{4q\log\Delta}$ neighbors of its own color. We set up an LLL instance for defective coloring very similarly to that for light partition previously:

\begin{itemize}
	\item The vertex set corresponds to that of the input graph ($\ents = \{_vA:v\in V\}$).
	\item The set of variables $\Vars$ will consist of one variable $color(v)$ for each $v\in V$, and this variable will be allocated to $_vA$ (i.e. $\vars(_vA) = \{color(v)\}$).
	\item These variables $color(v)$ are each uniformly distributed in $\{0,1\}$.
	\item The \emph{bad event} $_vA$ is that $v$ has at least $\frac{\Delta}{2}+\frac{\Delta}{4q\log\Delta}$  neighbors of its own color (so avoiding all bad events would imply a $(\frac{\Delta}{2},2q\log\Delta)$-defective coloring).
	\item Events $_vA\ne$ $_wA$ are therefore adjacent in $\aG$ iff $v$ and $w$ are adjacent in $G$ (i.e. $\ad=\Delta$).
	\item Events $_vA\ne$ $_wA$ are adjacent in $G^{\ents}$ iff $dist_G(v,w)\le 2$ (i.e. $d< 2\Delta^2$).
\end{itemize}  

Let $z_v$ denote the number of neighbors of $v$ that choose the same color as $v$ when the variables are sampled. The expectation of $z_v$ is $\mu:=deg(v)/2 \le \Delta/2$.  By a Chernoff bound, we obtain an upper bound $p$ on the probability of a bad event:

\begin{align*}
\Prob{z_v\ge \frac{\Delta}{2}+\frac{\Delta}{4q\log\Delta}}&\le \Prob{z_v\ge (1+\frac{\Delta}{4q\mu \log \Delta})\mu }\\
&\le e^{-\frac{\Delta^2 \mu}{3(4q\mu\log \Delta)^2}} = e^{-\frac{\Delta^2 }{48 \mu (q \log \Delta)^2}}\\
&\le e^{-\frac{\Delta }{24 (q \log \Delta)^2}}\enspace.
\end{align*}

So, our LLL instance has $p=e^{-\frac{\Delta }{24 (q \log \Delta)^2}}$ and $\ad = \Delta$. By Theorem \ref{thm:LLL}, therefore, it can be solved in $O((q \log \Delta)^2 + \log^{O(1)}\log n)$ rounds, so long as $(q \log \Delta)^2 = o(\frac{\Delta}{\log\Delta})$, i.e. $q=o(\sqrt\frac{\Delta}{\log^3\Delta})$.
\end{proof}

To color using more colors, we simply iterate this $2$-coloring procedure to iteratively subdivide the color classes.

\begin{lemma}\label{lem:def2}
For any $q\le \sqrt\frac{\Delta}{\log^4\Delta}$, a $(\Theta(q^2\log^4 q),q)$-defective coloring can be found in $O(q^2\log^3\Delta + \log\Delta\log^{O(1)}\log n)$ rounds of \rl, succeeding with high probability in $n$.
\end{lemma}

\begin{proof}
We iterate $\log\Delta - \log (q^2\log^4 q)$ times (our reason for requiring $q\le \sqrt\frac{\Delta}{\log^4\Delta}$ is so that this is positive). In each iteration, we use  Lemma \ref{lem:def1} to split each color class into two new color classes, resulting eventually in $2^{\log\Delta - \log (q^2\log^4 q)} = \frac{\Delta}{q^2\log^4 q}$ colors (as required for a $(q^2\log^4 q,q)$-defective coloring). We run the LLL instances for each of the current color classes simultaneously. This can be done since their induced graphs are entirely disjoint, and do not affect each other (no two nodes in different color classes can ever be colored the same color later). 

We will prove the following by induction: after $i$ iterations of Lemma \ref{lem:def1} in this way, the resulting coloring is a $(\frac{\Delta}{2^i},\frac 1i q\log\Delta)$-defective coloring.

As a base case, this claim is clearly true for $i=1$ since it is weaker than the statement of Lemma \ref{lem:def1}.

To prove the inductive step, in iteration $i$, we start with a $(\frac{\Delta}{2^{i-1}},\frac {1}{i-1}q\log\Delta)$-defective coloring and apply Lemma \ref{lem:def1} on each of the induced graphs of the current color classes. However, by the definition of an $(\frac{\Delta}{2^{i-1}},\frac {1}{i-1}q\log\Delta)$-defective coloring, the maximum degree within these induced graphs is $\Delta_i := \frac{\Delta}{2^{i-1}}+\frac{\Delta(i-1)}{2^{i-1} q\log\Delta}$, so we can use this quantity in place of $\Delta$ as the maximum degree in Lemma \ref{lem:def1}. Note that for any $i\le \log\Delta - \log (q^2\log^4 q)$, we have $\Delta_i \ge 2 q^2\log^4 q$, and therefore we indeed have $q=o(\sqrt\frac{\Delta_i}{\log^3\Delta'_i})$ as required in Lemma \ref{lem:def1}.

So, in iteration $i$, Lemma \ref{lem:def1} divides each color class into two (giving $2^i$ total colors) in such a way that every node has at most $\frac{\Delta_i}{2}+\frac{\Delta_i}{4q\log\Delta}$ neighbors of its color. Since

\begin{align*}
	\frac{\Delta_i}{2}+\frac{\Delta_i}{4q\log\Delta} &=
 \frac{\frac{\Delta}{2^{i-1}}+\frac{\Delta(i-1)}{2^{i-1} q\log\Delta}}{2}+\frac{\frac{\Delta}{2^{i-1}}+\frac{\Delta(i-1)}{2^{i-1} q\log\Delta}}{4q\log\Delta}\\
	&=
	\frac{\Delta}{2^i}+	\frac{\Delta(i-1)}{2^{i}q\log\Delta}
	+\frac{\Delta}{2^{i+1}q\log\Delta}+\frac{\Delta(i-1)}{2^{i+1}(q\log\Delta)^2}\\
	&\le
\frac{\Delta}{2^i}+	\frac{\Delta(i-1)}{2^{i}q\log\Delta}
+\frac{\Delta}{2^{i}q\log\Delta}\\
&=\frac{\Delta}{2^i}+	\frac{\Delta i}{2^iq\log\Delta}\enspace,
\end{align*}
this meets the criterion of a $(\frac{\Delta}{2^i},\frac 1i q\log\Delta)$-defective coloring, proving the claim by induction.

So, after $\log\Delta - \log (q^2\log^4 q)$ iterations, we have a $(q^2\log^4 q,\frac {1}{\log\Delta - \log (q^2\log^4 q)} q\log\Delta)$-defective coloring. This is therefore a $(q^2\log^4 q,q)$-defective coloring (since decreasing the second parameter weakens the requirement). Note that we have been assuming that the number of iterations $\log\Delta - \log (q^2\log^4 q)$ was an integer; if not, we simply round it, and this is what causes the statement of the lemma to give a $(\Theta(q^2\log^4 q),q)$-defective coloring rather than $(q^2\log^4 q,q)$.

We have run fewer than $\log\Delta$ iterations, each taking $O((q \log \Delta)^2 + \log^{O(1)}\log n)$ rounds of \rl by Lemma \ref{lem:def1}. The total round complexity is therefore $\log\Delta\cdot O((q \log \Delta)^2 + \log^{O(1)}\log n) = O(q^2\log^3 \Delta + \log\Delta \log^{O(1)}\log n)$. Each of these iterations succeeds with high probability in $n$, and therefore by taking a union bound over the failure probability in all iterations, we achieve high probability overall success.

\end{proof}

\subsection{Defective Edge Coloring}
A very similar application is that of defective edge coloring. A defective (vertex) coloring can be thought of as a vertex coloring such that the degree of the graph induced by each color class is low. We can similarly aim to color the \emph{edges} such that the degree of the graph induced by each color class is low:

\begin{definition}
	An $(x,q)$-defective edge coloring of a graph $G$ with maximum degree $\Delta$ is a $\frac{\Delta}{x}$-coloring of edges such that each node has fewer than $x + x/q$ adjacent edges of any color. 
\end{definition}

Again, we first find a defective edge coloring with two colors, and then iterate to increase the number of colors:

\begin{lemma}\label{lem:defe1}
	For any $q=o(\sqrt\frac{\Delta}{\log^3\Delta})$, a $(\frac{\Delta}{2},2q\log\Delta)$-defective edge coloring can be found in $O((q \log \Delta)^2 + \log^{O(1)}\log n)$ rounds of \rl, succeeding with high probability in $n$.
\end{lemma}

\begin{proof}
	A $(\frac{\Delta}{2},q\log\Delta)$-defective edge coloring is an edge coloring using only $2$ colors, such that every node has fewer than $\frac{\Delta}{2}+\frac{\Delta}{4q\log\Delta}$ adjacent edges of any color. We set up an LLL instance for defective edge coloring:
	
	\begin{itemize}
		\item The vertex set corresponds to the edge set of the input graph ($\ents = \{_eA:e\in E\}$).
		\item The set of variables $\Vars$ will consist of one variable $color(e)$ for each $e\in E$, and this variable will be allocated to $_eA$ (i.e. $\vars(_eA) = \{color(e)\}$).
		\item These variables $color(e)$ are each uniformly distributed in $\{0,1\}$.
		\item The \emph{bad event} $_eA$ is that either endpoint of $e$ has at least $\frac{\Delta}{2}+\frac{\Delta}{4q\log\Delta}$  adjacent edges of $e$'s color (so avoiding all bad events would imply a $(\frac{\Delta}{2},2q\log\Delta)$-defective edge coloring).
		\item Events $_eA\ne$ $_fA$ are therefore adjacent in $\aG$ iff $e$ and $f$ are adjacent in $G$ (i.e. $\ad=2\Delta-1$).
		\item Events $_eA\ne$ $_fA$ are adjacent in $G^{\ents}$ iff $dist_G(e,f)\le 2$ (i.e. $d< 4\Delta^2$).
	\end{itemize}  
	
Fix an edge $e=\{v,u\}$ and a particular endpoint $v$ of $e$ to consider. Let $z_e$ denote the number of (other) adjacent edges of $v$ that choose the same color as $e$ when the variables are sampled. The expectation of $z_v$ is $\mu \le \Delta/2$.  By a Chernoff bound, we obtain an upper bound $p$ on the probability that $v$ has more than $\frac{\Delta}{2}+\frac{\Delta}{5q\log\Delta}$ adjacent edges sharing $e$'s color:
	
	\begin{align*}
		\Prob{z_v\ge \frac{\Delta}{2}+\frac{\Delta}{5q\log\Delta}}&\le \Prob{z_v\ge (1+\frac{\Delta}{5q\mu \log \Delta})\mu }\\
		&\le e^{-\frac{\Delta^2 \mu}{3(5q\mu\log \Delta)^2}} = e^{-\frac{\Delta^2 }{75 \mu (q \log \Delta)^2}}\\
		&\le e^{-\frac{\Delta }{38 (q \log \Delta)^2}}\enspace.
	\end{align*}
	
Therefore, with probability least $1-e^{-\frac{\Delta }{38 (q \log \Delta)^2}}$, $v$ has fewer than $\frac{\Delta}{2}+\frac{\Delta}{5q\log\Delta}$ adjacent edges other than $e$ sharing $e$'s color, and so fewer than $\frac{\Delta}{2}+\frac{\Delta}{5q\log\Delta}+1< \frac{\Delta}{2}+\frac{\Delta}{4q\log\Delta}$ including $e$ itself. The same applies without loss of generality to the other endpoint $u$ of $e$. So, the total probability of the bad event $_eA$ is at most $2e^{-\frac{\Delta }{38 (q \log \Delta)^2}}$.
	
Then, our LLL instance has $p=2e^{-\frac{\Delta }{38 (q \log \Delta)^2}}$ and $\ad = 2\Delta-1 = O(\Delta)$. By Theorem \ref{thm:LLL}, therefore, it can be solved in $O((q \log \Delta)^2 + \log^{O(1)}\log n)$ rounds, so long as $(q \log \Delta)^2 = o(\frac{\Delta}{\log\Delta})$, i.e. $q=o(\sqrt\frac{\Delta}{\log^3\Delta})$.
\end{proof}

To reach the defective edge coloring we require for our subsequent application to $\Delta+o(\Delta)$ edge coloring, we iterate this $2$-coloring in the same way as for defective vertex coloring.

\begin{lemma}\label{lem:defe2}
	For any $q\le \sqrt\frac{\Delta}{\log^4\Delta}$, a $(\Theta(q^2\log^4 q),q)$-defective edge coloring can be found in $O(q^2\log^3\Delta + \log\Delta\log^{O(1)}\log n)$ rounds of \rl, succeeding with high probability in $n$.
\end{lemma}

The proof is essentially identical to that of Lemma \ref{lem:def2}, so we defer it to the appendix.

\section{$\Delta+o(\Delta)$ Edge Coloring}\label{sec:edgecolor}

We now reach our final application: edge coloring a graph using fewer colors than required by a greedy algorithm, as stated in Theorem \ref{thm:edgecolor}.

\begin{proof}[Proof of Theorem \ref{thm:edgecolor}]
Our aim is to employ the edge coloring algorithm of Chang et al. \cite{CHLPU19}, which solves the problem using a series of LLL calls to carefully control certain properties while coloring. Using existing LLL algorithms, though, results in round complexities that are either $poly(\Delta,\log\log n)$ (Fischer and Ghaffari's algorithm \cite{FG17}), or $poly(\log n)$ (Chung, Pettie, and Su's algorithm \cite{CPS17}). To combat this, we will apply our defective edge coloring result to reduce to instances with $poly(1/\eps)$ maximum degree, to which applying \cite{CHLPU19} armed with our LLL algorithm (or indeed that of \cite{FG17}) then takes only $poly(1/\eps, \log\log n)$ rounds complexity. This yields an improvement in overall round complexity for the most difficult regime when $\Delta$ is between roughly $\log\log n$ and $\log n$.

Firstly, note that the algorithm of Chang et al. \cite{CHLPU19} equipped with Fischer and Ghaffari's LLL algorithm \cite{FG17} (or our own Theorem \ref{thm:LLL}) already achieves $poly(1/\eps, \log\log n)$ round complexity when $1/\eps = \Delta^{\Omega(1)}$, since the resulting $poly(\Delta, \log\log n)$ round complexity is $poly(1/\eps, \log\log n)$. Secondly, the bound is also already met when $1/\eps = \Delta^{o(1)}$ and $\Delta \ge \log^{2}n$: in this case, the bad events occur with probability at most $exp(-\eps^2 \Delta/ \log^{4+o(1)} \Delta) = n^{-\omega(1)}$, and therefore all are avoided with high probability in $n$ upon initially sampling the variables, with no LLL algorithm required. The round complexity of \cite{CHLPU19} is then $O(\log(1/\eps) + \log^* n)$, which is much lower than $poly(1/\eps, \log\log n)$. So, it remains to prove the theorem for when $\Delta<\log^2 n$ and $1/\eps = \Delta^{o(1)}$, which we will henceforth assume.

We first employ a defective edge coloring in order to reduce the degree of the instances we need to consider. Using Lemma \ref{lem:defe2}, we obtain $(cq^2\log^4 q,q)$-defective edge coloring (for some constant $c$) of the input graph $G$, with $q=\eps^{-2}$ (noting that $q=\Delta^{o(1)} \le \sqrt\frac{\Delta}{\log^4\Delta}$ as required). This takes only $poly(1/\eps, \log\Delta, \log\log n) = poly(1/\eps, \log\log n)$ rounds of \rl. We call the resulting color classes \emph{buckets} (in order to avoid confusion with the final output coloring). The $(1+\eps)\Delta$ colors in the palette are divided lexicographically among the $\frac{\Delta}{cq^2\log^4 q}$ buckets, as equally as possible (i.e. with each bucket receiving either the floor or ceiling of $(1+\eps)\Delta\cdot \frac{cq^2\log^4 q}{\Delta}$ colors). Each edge in $E(G)$ will be colored with one of the colors in its bucket in the final output. This means that, from now on, we can treat the buckets separately and solve concurrently on their induced graphs, since no vertices in different buckets will ever cause a coloring conflict.

By the property of a defective coloring, the maximum degree within each bucket is now $\Delta':= cq^2\log^4 q + cq\log^4 q$. We define $\eps':= \frac{\eps}{2}$. The number of colors in each bucket is at least 

\begin{align*}
	(1+\eps)\Delta\cdot \frac{cq^2\log^4 q}{\Delta} - 1 &= cq^2\log^4 q + \eps cq^2 \log^4 q- 1\\
	&= cq^2\log^4 q +  cq^{1.5} \log^4 q- 1\\
	&\ge cq^2\log^4 q + \frac 12  cq^{1.5}\log^4 q +cq\log^4 q +   \frac{1}{2} cq^{0.5}\log^4 q \\
	&= (1+\frac{\eps}{2})(cq^2\log^4 q + cq\log^4 q)\\
	&= (1+\eps')\Delta' \enspace.
\end{align*}

We have now reduced the edge coloring instance to a collection of new instances with distinct palettes, each of which can therefore be solved concurrently. The maximum degree of each instance is $\Delta'=\Theta(q^2\log^4 q) = poly(1/\eps)$, and the number of colors available is $(1+\eps')\Delta'$, where $\eps'=q^{-1/2}/2 = \omega(\frac{\log^{2.5} \Delta'}{\sqrt{\Delta'}})$. So, we can apply the edge coloring algorithm of Chang et al. \cite{CHLPU19}, again equipped with the LLL algorithm of Theorem \ref{thm:LLL} or Fischer and Ghaffari \cite{FG17}, to solve each instance concurrently in $poly(\Delta',\log\log n) = poly(1/\eps, \log\log n)$ rounds. The total round complexity is therefore $poly(1/\eps,\log\log n)$.
\end{proof}

Theorem \ref{thm:edgecolor} clearly implies Corollary \ref{cor:edge}, though it only gives the best number of colors known for $\Delta$ up to around $\log n$; for higher $\Delta$, edge colorings with fewer colors were already known by using \cite{CHLPU19} either with the LLL algorithm of \cite{CPS17} or setting $\eps$ such that no LLL algorithm is required.

\begin{proof}[Proof of Corollary \ref{cor:edge}]
When $\Delta\le \log^8 n$, $(1+\eps)\Delta$ edge coloring can be performed in $\log^{O(1)}\log n$ rounds for any $\eps$ which is both $\omega(\frac{\log^{2.5}\Delta}{\sqrt{\Delta}})$ and $\log^{-O(1)}\log n$, by Theorem \ref{thm:edgecolor}. When $\Delta > \log^8 n$, $(1+\eps)\Delta$ edge coloring with any $\eps$ which is both $\log^{-O(1)}n$ and at least $\Delta^{-1/8}$ can be performed using the algorithm of Chang et al. \cite{CHLPU19} directly, since this is again in the regime where all bad events are avoided with high probability, with no LLL algorithm required. Then, the round complexity of \cite{CHLPU19} is $O(\log(1/\eps)+\log^* n) = O(\log\log n)$. 

This shows that $\Delta+o(\Delta)$ edge coloring can be performed in $ O(\log\log n)$ rounds over all ranges of $\Delta$. The best number of colors purely as a function of $\Delta$ that we obtain in this round complexity is $\Delta+\frac{\Delta}{\log^{O(1)} \Delta}$, though the bottleneck is only when $\Delta= \log^{\Theta(1)}n$, and for $\Delta$ outside this range fewer colors can be used.
\end{proof}

\section{Conclusions and Open Questions}
We have shown improved algorithms for the Lov\'asz Local Lemma, and an application to edge coloring, demonstration that graphs of maximum degree $\Delta$ can be $\Delta+o(\Delta)$ edge colored in $\log^{O(1)}\log n$ rounds of \rl. There well may be other applications of our results: in particular, the leading algorithm for $o(\Delta)$-coloring triangle-free graphs \cite{PS15} is similarly built on repeated calls to the Lov\'asz Local Lemma. We expect that adapting the problem to apply our LLL algorithm for resilient instances would improve the round complexity, and aim to do so in future work.

Many questions regarding the LLL still remain open. On graphs with $d=\log^{\omega(1)}\log n$, there is still a large gap between upper and lower bounds. LLL with polynomially-weakened criterion, for example, appears to become more difficult as $d$ increases (at least up to $\log^{1-\Omega(1)} n$), but the $\Omega(\log\log n)$-round lower bound of Brandt et al. \cite{BF+16} is proven on constant-degree graphs and does not increase with $d$. One can also ask for how strong an LLL criterion can we find a $\log^{O(1)}\log n$-round distributed algorithm over the whole range of $\Delta$. We have shown that the criterion $p\le \min \{d^{-c}, 2^{\frac{-d}{\log^{O(1)}\log n}}\}$ suffices (for some constant $c$), and it was already known \cite{CPS17} that $p \le \frac{1}{ed^2}\cdot 2^{-\frac{\log n}{\log^{O(1)}\log n}}\}$ also suffices. However, these criteria are a long way from the strongest criteria under which the LLL is solvable, and no $\log^{\omega(1)}\log n$-round lower bound is known for any solvable criterion.

The picture for edge coloring is also far from complete: there remain large gaps between upper and lower bounds on distributed complexity for most palette size regimes. For $2\Delta-1$-coloring, $\log^{O(1)}\log n$-round algorithms are known \cite{FGK17,GHK18,Harris19}, but the only lower bound is $\Omega(\log^* n)$ \cite{Linial92,Naor91}. Below $2\Delta-1$ colors, there is an $\Omega(\log_\Delta\log n)$-round lower bound \cite{CHLPU19}, and Corollary \ref{cor:edge} closes the corresponding upper bound to only a polynomial gap for some $\Delta+o(\Delta)$ number of colors\footnote{The precise amount depends on how $\Delta$ relates to $n$; see the proof of Corollary \ref{cor:edge}}. For even fewer colors, though, there is again a wide gap, with the best upper bound being the $poly(\Delta,\log n)$-round deterministic algorithm for $\Delta+1$ edge coloring of Bernshteyn \cite{Bernshteyn22}.

\section{Acknowledgements}
The author would like to thank Merav Parter and Artur Czumaj for illuminating discussions regarding the Lov\'asz Local Lemma.
\appendix

\bibliographystyle{alpha}

\bibliography{LLL}
\section{Appendix: Proof of Lemma \ref{lem:defe2}}
\begin{proof}[Proof of Lemma \ref{lem:defe2}]
	We again iterate $\log\Delta - \log (q^2\log^4 q)$ times as in proof of Lemma \ref{lem:defe2}. In each iteration, we use  Lemma \ref{lem:def1} to split each color class into two new color classes, resulting eventually in $2^{\log\Delta - \log (q^2\log^4 q)} = \frac{\Delta}{q^2\log^4 q}$ colors. We run the LLL instances for each of the current color classes simultaneously. This can be done since their induced graphs are entirely disjoint, and do not affect each other (no two edges in different color classes can ever be colored the same color later, and therefore the bad event at each edge depends only on the variables of edges in its current color class). 
	
	We will prove the following by induction: after $i$ iterations of Lemma \ref{lem:defe1} in this way, the resulting coloring is a $(\frac{\Delta}{2^i},\frac 1i q\log\Delta)$-defective edge coloring.
	
	As a base case, this claim is clearly true for $i=1$ since it is weaker than the statement of Lemma \ref{lem:defe1}.
	
	To prove the inductive step, in iteration $i$, we start with a $(\frac{\Delta}{2^{i-1}},\frac {1}{i-1}q\log\Delta)$-defective edge coloring and apply Lemma \ref{lem:def1} on each of the induced graphs of the current color classes. However, by the definition of an $(\frac{\Delta}{2^{i-1}},\frac {1}{i-1}q\log\Delta)$-defective edge coloring, the maximum degree within these induced graphs is $\Delta_i := \frac{\Delta}{2^{i-1}}+\frac{\Delta(i-1)}{2^{i-1} q\log\Delta}$, so we can use this quantity in place of $\Delta$ as the maximum degree in Lemma \ref{lem:def1}. Note that for any $i\le \log\Delta - \log (q^2\log^4 q)$, we have $\Delta_i \ge 2 q^2\log^4 q$, and therefore we indeed have $q=o(\sqrt\frac{\Delta_i}{\log^3\Delta'_i})$ as required in Lemma \ref{lem:def1}.
	
	So, in iteration $i$, Lemma \ref{lem:defe1} divides each color class into two (giving $2^i$ total colors) in such a way that every node has at most $\frac{\Delta_i}{2}+\frac{\Delta_i}{4q\log\Delta}$ adjacent edges of each color. Since
	
	\begin{align*}
		\frac{\Delta_i}{2}+\frac{\Delta_i}{4q\log\Delta} &=
		\frac{\frac{\Delta}{2^{i-1}}+\frac{\Delta(i-1)}{2^{i-1} q\log\Delta}}{2}+\frac{\frac{\Delta}{2^{i-1}}+\frac{\Delta(i-1)}{2^{i-1} q\log\Delta}}{4q\log\Delta}\\
		&=
		\frac{\Delta}{2^i}+	\frac{\Delta(i-1)}{2^{i}q\log\Delta}
		+\frac{\Delta}{2^{i+1}q\log\Delta}+\frac{\Delta(i-1)}{2^{i+1}(q\log\Delta)^2}\\
		&\le
		\frac{\Delta}{2^i}+	\frac{\Delta(i-1)}{2^{i}q\log\Delta}
		+\frac{\Delta}{2^{i}q\log\Delta}\\
		&=\frac{\Delta}{2^i}+	\frac{\Delta i}{2^iq\log\Delta}\enspace,
	\end{align*}
	this meets the criterion of a $(\frac{\Delta}{2^i},\frac 1i q\log\Delta)$-defective edge coloring, proving the claim by induction.
	
	So, after $\log\Delta - \log (q^2\log^4 q)$ iterations, we have a $(q^2\log^4 q,\frac {1}{\log\Delta - \log (q^2\log^4 q)} q\log\Delta)$-defective edge coloring. This is therefore a $(q^2\log^4 q,q)$-defective edge coloring (since decreasing the second parameter weakens the requirement). Again, the lemma statement is weakened to a $(\Theta(q^2\log^4 q),q)$-defective edge coloring since $\log\Delta - \log (q^2\log^4 q)$ may not have been an integer.
	
	We have run fewer than $\log\Delta$ iterations, each taking $O((q \log \Delta)^2 + \log^{O(1)}\log n)$ rounds of \rl by Lemma \ref{lem:defe1}. The total round complexity is therefore $\log\Delta\cdot O((q \log \Delta)^2 + \log^{O(1)}\log n) = O(q^2\log^3 \Delta + \log\Delta \log^{O(1)}\log n)$. Each of these iterations succeeds with high probability in $n$, and therefore by taking a union bound over the failure probability in all iterations, we achieve high probability overall success.
	
\end{proof}
\end{document}